\documentclass[11pt]{article}

\usepackage{makeidx}
\usepackage{euscript}
\usepackage{amsmath,amssymb,amsfonts}
\usepackage{dsfont}
\usepackage{amsmath, amssymb}
\usepackage{latexsym}
\usepackage{subfigure}
\usepackage{graphicx}
\usepackage{url}
\usepackage{times}
\usepackage[scaled=0.92]{helvet} 
\usepackage{euler}
\usepackage{fancybox}
\usepackage{fullpage}
\usepackage{hyperref}
\usepackage[noend]{algorithmic}
\usepackage{algorithm}
\usepackage{color}
\usepackage{enumitem}
\usepackage{tcolorbox}

\usepackage{wrapfig}

\newenvironment{proof}{\noindent{\bf Proof : \ }}{\hfill$\Box$\par\medskip}

\newtheorem{theorem}{Theorem}
\newtheorem{corollary}[theorem]{Corollary}
\newtheorem{lemma}[theorem]{Lemma}

\newtheorem{definition}[theorem]{Definition}

\newtheorem{claim}[theorem]{Claim}

\newenvironment{proofof}[1]{\begin{trivlist} \item {\bf Proof
#1:~~}}
  {\qed\end{trivlist}}

\renewcommand{\Pr}[1]{\ensuremath{\mathbf{Pr}[#1]}}

\renewcommand{\paragraph}[1]{ \noindent {\bf #1}}

\newcommand{\COMMENTED}[1]{{}}
\newcommand{\NAT}{\ensuremath{\mathbb{N}}}
\newcommand{\NATURAL}{\NAT}

\newcommand{\Ex}[1]{\ensuremath{\mathbf{E}[#1]}}

\newcounter{proccnt}

\date{}

\newlength{\savedparindent}
\newcommand{\SaveIndent}{\setlength{\savedparindent}{\parindent}}
\newcommand{\RestoreIndent}{\setlength{\parindent}{\savedparindent}}

\newcommand{\InGrayMiddle}[1]{%
\SaveIndent{} %
\centerline{ \fcolorbox[rgb]{0,0,0}{0.95,0.95,0.95}{
\begin{minipage}{0.85\linewidth} %
\RestoreIndent{}%
#1
\end{minipage}
} } }

\title{Dynamic Maximal Independent Set}

\author{
Morteza Monemizadeh\thanks{Work was done while the author was at Amazon AI, Palo Alto, CA, USA.  Email: {\tt m.monemizadeh@gmail.com}.}
}

\begin{document}

\sloppy
\setlength{\abovecaptionskip}{0.1ex}
 \setlength{\belowcaptionskip}{0.1ex}
 \setlength{\floatsep}{0.1ex}
 \setlength{\textfloatsep}{0.1ex}

 \abovedisplayskip.30ex
   \belowdisplayskip.30ex
   \abovedisplayshortskip.30ex
   \belowdisplayshortskip.30ex

\maketitle\thispagestyle{empty}

\begin{abstract}
Given a stream $\mathcal{S}$ of insertions and deletions of edges of an underlying graph $G$ (with fixed vertex set $V$ where $n=|V|$ is the number of vertices of $G$), 
we propose a dynamic algorithm that  maintains a maximal independent set (MIS) of $G$ (at any time $t$ of the stream $\mathcal{S}$) 
with amortized update time $O(\log^3 n)$. 
\end{abstract}



\section{Introduction}
Very recently at STOC 2018, Assadi,  Onak, Schieber, and Solomon \cite{DBLP:conf/stoc/AssadiOSS18} proposed 
a deterministic dynamic algorithm for maintaining a
maximal independent set (MIS) with amortized update time $O(min(\Delta,m^{3/4}))$, where $\Delta$ is a fixed bound on the maximum degree 
in the graph and $m$ is the (dynamically changing) number of edges. 
Later, Gupta and Khan \cite{DBLP:journals/corr/abs-1804-01823} and independently, Du and Zhang \cite{DBLP:journals/corr/abs-1804-08908} presented deterministic algorithms for 
dynamic MIS with update times of $O(m^{2/3})$ and $O(m^{2/3}\cdot \sqrt{
\log m})$, respectively. Du and Zhang also gave a randomized algorithm with update time $\tilde{O}(\sqrt{m})$.
Later at SODA 2019, Assadi,  Onak, Schieber, and Solomon \cite{DBLP:conf/soda/AssadiOSS19} developed 
the first fully dynamic (randomized) algorithm for maintaining a MIS with $\min(\tilde{O}(\sqrt{n}),\tilde{O}(m^{1/3}))$ expected amortized update time. 

Here we develop the first randomized dynamic algorithm for MIS with amortized update time $O(\log^3 n)$. 
Our main result is stated in the following theorem. 

\begin{theorem}
\label{thm:dynamic:mis}
Let $S$ be a stream of insertions and deletions of edges of an underlying unweighted graph $G$ with a fixed vertex set $V$ of size $n=|V|$. 
Then, there exists a randomized dynamic algorithm that maintains a maximal independent set of $G$ using amortized $O(\log^3 n)$ update time. 
\end{theorem}

\paragraph{Overview of Algorithm.}
To prove this theorem we first devise an offline MIS algorithm in Section \ref{sec:mis} and then 
in Section \ref{sec:insert:delete} we show how to implement steps of this offline algorithm in a streaming fashion 
while maintaining a maximal independent set using fast update time. 
In Section \ref{sec:main:alg} we give our dynamic algorithm that handles insertions and deletions. 
First we explain the offline algorithm. 

Let $G(V,E)$ be an undirected unweighted graph with $n = |V|$ vertices and $m = |E|$ edges. 
We consider $k$ epochs during which we build levels $L_1, \cdots, L_k$ of independent vertices for $k = O(\log n)$. 
At the beginning of epoch $i$ we assume we have a graph $G_i(V_i,E_i)$. 
For the first epoch, we let $G_1(V_1, E_1) = G(V,E)$. 

At epoch $i$, we repeat the following sampling process for $s = O(\frac{n_i^2}{m_i})$ times: 
Repeat sampling (with replacement) a vertex $w \in V_i$ uniformly at random as long as $I_i \cup \{w\}$ is not an independent set in $G_i$. 
Once we sample a vertex $w$ for which $I_i \cup \{w\}$ is an independent set, 
we then let $I_i = I_i \cup \{w\}$ and $N(I_i) = N(I_i)  \cup N_{G_i}(w)$. 
If the graph $G_i$ is sparse (i.e., $|E_I| \le |V_i|$), we sample all vertices and the ordered set 
$S_i$ will be a random shuffle of vertices of $V_i$. 


We let $N(I_i)$ be the set of neighbors of $I_i$ in the graph $G_i(V_i,E_i)$ and we remove $I_i$ and $N(I_i)$ from $G_i(V_i,E_i)$. 
The level $L_i$ consists of the vertex set $V_i$,  the multiplicative inverse or reciprocal for the average degree of $G_i$ which is 
$\frac{|V_i|}{|E_i|}$, and the independent set $I_i$ and its neighbor set $N(I_i)$.  
We remove $I_i$ and $N(I_i)$ from the graph $G_i$ and recursively start the next epoch $E_{i+1}$. 

Next we explain the idea behind our edge insertion and deletion subroutines. 
Suppose we have a level set $\mathcal{L} = \cup_{i=1}^k L_i$ of $k$ levels of an underlying graph $G(V,E)$ where 
each level $L_i$ is a quadruple $L_i = (V_i, \frac{|V_i|}{|E_i|}, I_i, N(I_i))$ and  $\cup_{i=1}^k I_i$ is a MIS of $G$. 

Suppose we want to add an arbitrary edge $e = (u,v)$ to $G$. Let $G'= G(V,E \cup \{e\})$. 
The amount of recomputation that the insertion of an edge $e=(u,v)$ imposes while reconstructing a MIS of $G'$ given the current MIS of $G$ 
depends on where this edge is being inserted. 
We consider two types of insertions, \emph{heavy} insertions and \emph{light} insertions. 
Roughly speaking, an insertion is a heavy insertion if it changes the current maximal independent set; otherwise it is a light insertion. 
We show that heavy insertions are rare and the majority of insertions 
are in fact light insertions for which we do not need to do significant (re)-computation. 
So, we can use the budget that light insertions provides to us for heavy insertions. 

Intuitively, we have the following observation. 
At an epoch $i$, the independent set $I_i$ has $\Theta(\frac{n_i^2}{m_i})$ vertices, 
the cut $(I_i, N(I_i))$ consists of $\Theta(n_i)$ edges and $G_i$ contains $m_i$ edges. 
So, for any change in the independent set $I_i$, the adversary needs to update $\Theta(\frac{m_i^2}{n_i^2})$ edges of the graph $G_i(V_i,E_i)$. 
As an example, if $G_i$ has $n_i$ vertices and $n_i\sqrt{n_i}$ edges, then $|I_i|$ has $\Theta(\sqrt{n_i})$ vertices, 
the adversary needs to update $\Theta(n_i)$ edges in order to change $I_i$.
The same happens for edge deletions.

\subsection{Preliminaries} 
Let $G(V,E)$ be an undirected unweighted graph with $n=|V|$ vertices and $m=|E|$ edges. 
We assume that there is a unique numbering for the vertices in
$V$ so that we can treat $v \in V$ as a unique number $v$ for $1 \le v \le n=|V|$.
We denote an edge in $E$ with two endpoints $u,v\in V$ by
$(u,v)$. The graph $G$ can have at most ${n \choose 2} = n(n-1)/2$ edges.
Thus, each edge can also be thought of as referring to a unique number
between 1 and ${n \choose 2}$. Here $[x]=\{1,2,3,\cdots,x\}$ when $x\in \NATURAL$.

Given a vertex $v \in V$ we let $N_G(v) = \{u\in V: (u,v) \in E\}$ be the neighborhood of $v$. 
We let $d_G(v) = |N_G(v)|$ be the degree of the vertex $v$. When it is clear from the context we often drop $G$ from $d_G(v)$ and $N_G(v)$ 
and simply write them as $d(v)$ and $N(v)$. The average degree of the graph $G$ is $d(G) = \frac{1}{n}\cdot \sum_{v \in V} d_G(v)$. \\

Next we define a maximal independent set. 

\begin{definition}[Maximal Independent Set (MIS)]
\label{def:mis}
Given an undirected Graph $G(V,E)$, an independent set is a subset of nodes $U \subseteq V$, 
such that no two nodes in $U$ are adjacent. An independent set is maximal if no node can be added without
violating independence. 
\end{definition}

There is a simple greedy algorithm that reports a MIS of $G$. 
In particular, we scan the nodes of $G$ in arbitrary order. 
If a node $u$ does not violate independence, we add $u$ to the MIS. If $u$ violates
independence, we discard $u$. \\

\paragraph{Dynamic Model.}
Let $S$ be a stream $S$ of insertions and deletions of edges. 
We define time $t$ to be the $t$-th operation (i.e., insertion or deletion) of stream $S$. 
Let $I_t$ be a maximal independent set of an underlying graph $G_t(V,E_t)$ whose edge set $E_t$ is the set of edges that are inserted up to time $t$ but not deleted. 
The update time of a dynamic algorithm $\mathcal{A}$ is the time that $\mathcal{A}$ needs to compute a MIS $I_t$ 
of graph $G_{t}(V,E_{t})$ given a MIS $I_{t-1}$ of graph $G_{t-1}(V,E_{t-1})$.  
The update time can be worst-case or amortized. \\


\paragraph{Query Model.}
We assume the input graph $G(V,E)$ is represented as an adjacency list. 
We could also assume that $G$ is represented as an adjacency matrix, but adjacency matrices are often suitable for dense graphs where $|E| = \Theta(|V|^2)$. 
In dynamic scenarios we may end up with many edge deletions so that the graph become very sparse for which the adjacency matrix representation may not be appropriate. 
The complexity of dynamic algorithms for graph problems is often measured based on number of neighbor queries 
where for every vertex $v \in V$, we query its $i$-th neighbor. 
We assume that a neighbor querie takes constant time. Therefore, querying the full neighborhood of a vertex $v\in V$ takes $O(d_G(V))$ time. 
We let $\mathcal{Q}(\mathcal{A},G)$ be the number of neighbor queries that an algorithm $\mathcal{A}$ makes to compute a function. \\

%
In this paper, we use the following concentration bound. 

\begin{lemma}[Additive Chernoff Bound]\cite{Cher}
\label{lem:cher}
Let $Y_1,\cdots,Y_m$ denote $m$ identically  distributed and independent random variables such that 
$\Ex{Y_i} =p$ for $1\leq i\leq n$ for a fixed $0\leq p\leq 1$. Let $0<t<1,t\geq p$. 
For $Y=\sum_{i=1}^m Y_i$ it holds that 
\[
  \Pr{Y\geq t\cdot m}\leq \left[\left(\frac{p}{t}\right)^t\cdot \left(\frac{1-p}{1-t}\right)^{(1-t)}\right]^m. 
\]
\end{lemma}

\section{Maximal Independent Set (MIS)}
\label{sec:mis}

The pseudocode of our offline MIS algorithm is given in Algorithm (1) {\sf Maximal-Independent-Set}. 


\begin{algorithm*}
\label{alg:offline:mis}
\noindent
\textbf{Input:} Unweighted undirected graph $G(V,E)$ with $n=|V|$ vertices and $m=|E|$ edges.

\begin{algorithmic}[1]
	\STATE Let $i=0$ and $G_{i}(V_{i},E_{i}) = G(V,E)$. Let $c = 34$.
	\WHILE{$V_i \neq \emptyset$}
		\STATE Let $j = 0$, $I_i = N(I_i)= \emptyset$, $n_i = |V_i|$, $m_i=|E_i|$, and $t = \frac{n_i^2}{c\cdot m_i}$. 
		\WHILE{$j \le \max( t ,1)$}
			 \WHILE{ TRUE }
				 \STATE Sample a vertex  $v \in V_r$ uniformly at random.  
				 \IF{$I_r \cup \{v\}$ is an independent set in the induced graph of $V_r$}
				 	\STATE Break the true while loop.
				 \ENDIF
			 \ENDWHILE
		
			\STATE Let $I_i = I_i \cup \{v\}$, $N(I_r) = N(I_r) \cup N_{G_R}(v)$ and $j = j+1$. 
		\ENDWHILE
		\STATE Let level $L_i$ be the quadruple $(V_i, \frac{n_i}{m_i}, I_i, N(I_i))$.
		\STATE Let $G_{i+1}(V_{i+1},E_{i+1})$ be the indued subgraph on $V_{i+1} = V_i \backslash (I_i \cup N(I_i))$.
		\STATE Let $i = i + 1$. 
	\ENDWHILE
\end{algorithmic}

\noindent
\textbf{Output:} Return the level set $ \mathcal{L} = \cup_{i=1}^k L_i$ where $k= |\mathcal{L}| = O(\log n)$ is the number of levels.

\caption{Maximal-Independent-Set}
\end{algorithm*}

%
%
%
%

The MIS algorithm Algorithm (1) {\sf Maximal-Independent-Set} is the same as the following MIS algorithm. 
Let $G(V,E)$ be an undirected unweighted graph with $n = |V|$ vertices and $m = |E|$ edges. 
We consider $k$ epochs during which we build levels $L_1, \cdots, L_k$ of independent vertices for $k = O(\log n)$. 
At the beginning of epoch $i$ we assume we have a graph $G_i(V_i,E_i)$. 
For the first epoch, we let $G_1(V_1, E_1) = G(V,E)$. 

At epoch $i$, we sample an ordered set $S_i \subseteq V_i$ of vertices with probability $ \frac{|V_i|}{|E_i|}$  
and we let $I_i$ be a MIS that we find greedily for the induced sub-graph $H(S_i,E[S_i])$. 
If the graph $G_i$ is sparse (i.e., $|E_I| \le |V_i|$), we sample all vertices and the ordered set 
$S_i$ will be a random shuffle of vertices of $V_i$.

We let $N(I_i)$ be the set of neighbors of $I_i$ in the graph $G_i(V_i,E_i)$ and we remove $I_i$ and $N(I_i)$ from $G_i(V_i,E_i)$. 
The level $L_i$ consists of the vertex set $V_i$,  the multiplicative inverse or reciprocal for the average degree of $G_i$ which is 
$\frac{|V_i|}{|E_i|}$, and the independent set $I_i$ and its neighbor set $N(I_i)$.  
We remove $I_i$ and $N(I_i)$ from the graph $G_i$ and recursively start the next epoch $E_{i+1}$. 

The pseudocode of this algorithm is given Algorithm (2) {\sf Maximal-Independent-Set (Subset-Sampling)}. 

\begin{algorithm*}
\label{alg:offline:mis:2nd:try}
\noindent
\textbf{Input:} Unweighted undirected graph $G(V,E)$ with $n=|V|$ vertices and $m=|E|$ edges.

\begin{algorithmic}[1]
	\STATE Let $i=1$ and $G_{i}(V_{i},E_{i}) = G(V,E)$.
	\WHILE{$V_i \neq \emptyset$}
		\STATE Let $S_i$ be a sample set where each vertex $v \in V_i$ is sampled with probability $\Pr{v}=\min(\frac{|V_i|}{|E_i|},1)$. 
		\STATE Let $H(S_i , E[S_i])$ be the induced subgraph of $S_i$, where $E[S_i] = \{(u,v) \in E_i: u \in S_i \text{ and } v \in S_i\}$.
		\STATE Let $I_i$ be the output MIS of the greedy MIS for the graph $H(S_i,E[S_i])$.
		\STATE Let $N(I_i)=\{v\in V_i \backslash I_i: \exists u\in I_i \text{ and }  (u,v) \in H(S_i,E[S_i])\}$ be the neighbor set of $I_i$.
		\STATE Let level $L_i$ be the quadruple $(V_i, \frac{|V_i|}{|E_i|},  I_i, N(I_i))$.
		\STATE Let $G_{i+1}(V_{i+1},E_{i+1})$ be the indued subgraph on $V_{i+1} = V_i \backslash (I_i \cup N(I_i))$.
		\STATE Let $i = i + 1$. 
	\ENDWHILE
\end{algorithmic}

\noindent
\textbf{Output:} Return the levels $\mathcal{L} = \cup_{i =1}^k L_i$ where $k$ is the number of levels.  

\caption{Maximal-Independent-Set (Subset-Sampling)}
\end{algorithm*}


\subsection{Analysis}


First we prove that the induced subgraph $H(S_i , E[S_i])$ is sparse,  that is,  $|E[S_i]| \le c \cdot |S_i|$ for a constant $c \ge 1$. 

\begin{lemma}
\label{lem:bound:sample:size}
Let $G_i(V_i,E_i)$ be an undirected unweighted graph at the beginning of epoch $i$ of Algorithm \ref{alg:offline:mis}. 
Assume that $|E_i| > |V_i|$. 
With probability at least $1/2$, the number of vertices in $S_i$ is $|S_i| \ge \frac{n_i^2}{4m_i} \ge \frac{|E[S_i]|}{16}$ .
\end{lemma}

\begin{proof}
Let $n_i = |V_i|$ be the number of vertices in $V_i$ and $m_i = |E_i|$ be the number of edges in $E_i$. 
Suppose the vertices in $V_i$ are $v_1,\cdots, v_{n_i}$. 
Corresponding to the vertex $v_j$ we define an indicator random variable $X_j$  for the event that $v_j$ is sampled. 
We define a random variable $X = \sum_{j \in [n_i]} X_j$.  
Since $\Ex{X_j} = \Pr{X_j} = \frac{|V_i|}{|E_i|} = \frac{n_i}{m_i}$, we have $\Ex{X} = \frac{n_i^2}{m_i}$. 
Using Markov Inequality, $\Pr{X \ge \frac{1}{4}\cdot \frac{n_i^2}{m_i}} \le 1/4$. 

Next suppose the edges in $E_i$ are $e_1,\cdots, v_{m_i}$. 
Corresponding to the edge $e_j = (u_j,v_j)$ we define an indicator random variable $Y_j$  for the event that $E_j$ is in $E[S_i]$. 
We define a random variable $Y = \sum_{j \in [m_i]} Y_j$.  
Since $\Ex{Y_j} = \Pr{Y_j} = \Pr{u_j,v_j \in S_i} = \Pr{u_j\in S_i} \cdot \Pr{v_j\in S_i}  = (\frac{|V_i|}{|E_i|})^2 = (\frac{n_i}{m_i})^2$. 
We then have $\Ex{Y} = \frac{n_i^2}{m_i}$. 
Using Markov Inequality, $\Pr{X \ge 4\cdot \frac{n_i^2}{m_i}} \le 1/4$. 

Thus, using the union bound, with probability at least $1/2$, $|S_i| \ge \frac{n_i^2}{4m_i} \ge \frac{|E[S_i]|}{16}$. 

%
%
%
%
\end{proof}

Now we prove that the independent set $I_i$ reported at Epoch $i$ of Algorithm \ref{alg:offline:mis} is relatively big with respect 
to the sampled set size $S_i$ and also a constant fraction of vertices in $V_i$ are neighbors of $I_i$ that can be removed 
once we recurse the sampling process for the graph $G_{i+1}$. 

\begin{lemma}
\label{lem:big:independent:Set}
Let $H(S_i , E[S_i])$ be the induced subgraph  reported at Epoch $i$ of Algorithm \ref{alg:offline:mis}.
Then, the random greedy algorithm for the maximal independent set problem returns an independent set $I_i$ 
of size $|I_i| \ge \frac{|S_i|}{34}$. 
\end{lemma}

\begin{proof} 
First we find the lower bound on the size of the independent set $I_i$. 
Using Lemma \ref{lem:bound:sample:size} we have $|S_i| \ge \frac{n_i^2}{4m_i} \ge \frac{|E[S_i]|}{16}$.  
Therefore, the average degree of $H(S_i , E[S_i])$ is upper bounded by $d(G_i) = \frac{|E[S_i]|}{|S_i|} \le 16$ 
which means that the independent set $I_i$ is of size $|I_i| \ge \frac{|S_i|}{2(d(G_i)+1)} = \frac{|S_i|}{34}$. 
\end{proof}

%
%
%

\begin{lemma}
\label{lem:neighbor:IS:large}
Let $I_i$ be the independent set in the graph $G_i(V_i,H_i)$ that is reported by Algorithm \ref{alg:offline:mis}.  
Let $N(I_i) = \{  v \in V_i | \exists u \in I_i :  (u,v) \in E_i\}$ be the set of vertices in $G_i$ that are neighbors of $I_i$. 
Then, we have $\Pr{|N(I_i)| \ge \frac{n_i}{900}} \ge 2/3$. 
\end{lemma}

\begin{proof}
Let us consider the independent set $I_i = \{u_1, \cdots, u_{t}\}$ in the graph $G_i(V_i,E_i)$ where $t = \frac{|V_i|^2}{34 \cdot |E_i|}$. 
Suppose when we sample the vertex $u_j$, the set $N(I_i^{j-1})$ is the set of vertices of $G_i$ that are neighbor 
to one of the vertices $u_1,\cdots, u_{j-1}$. That is, $N(I_i^{j-1}) = \{ v \in V_i | \exists 1 \le \ell \le j-1: (v,u_{\ell}) \in E_i\} $. 
Assume that $|N(I_i^{j-1})| < |V_i| / 2$; otherwise, removing the pair set $(I_i, N(I_i))$ from the graph $G_i$ drops 
the number of vertices by half and we can recurse with the induced subgraph of the remaining vertex set. 

Now suppose we sample the vertex $u_j$. 
We define a random variable $X_j$ for the number of vertices in $V_i \backslash N(I_i^{j-1})$ that are neighbors of $u_j$. 
In expectation we have $\Ex{X_j} = d_{G_i}(u_j) \cdot \frac{|V_i \backslash N(I_i^{j-1})|}{|V_i|}$ where $d_{G_i}(u_j)$ is the degree of $u_j$ in $G_i$. 
Let us define a random variable $X = \sum_{j =1}^{t} X_j$. 
We then have 
\[
	\Ex{X} = \sum_{j =1}^{t} \Ex{X_j} = \sum_{j =1}^{t} d_{G_i}(u_j) \cdot \frac{|V_i \backslash N(I_i^{j-1})|}{|V_i|} \ge \frac{1}{2} \cdot  \sum_{j =1}^{t} d_{G_i}(u_j) \enspace .
\]	

Now corresponding to the vertex $u_j$ we define a random variable $Y_j$  for the degree of $u_j$. 
We also define a random variable $Y = \sum_{j \in [t]} Y_j$.  Observe that $\Ex{Y_j} = \frac{m_i}{n_i}$. 
Therefore, we have 
$$\Ex{Y} = t \cdot \frac{m_i}{n_i} \ge  \frac{n_i^2}{34 \cdot m_i} \cdot \frac{m_i}{n_i} = \frac{n_i}{34} \enspace . $$

We then apply Markov Inequality to obtain  
$$
	\Pr{ Y \le \frac{n_i}{136}} = \Pr{\sum_{u_j \in I_i} d_{G_i} (u_j) \le \frac{n_i}{136}} \le 1/4 \enspace .
$$ 

Therefore, with probability at least $3/4$, $\sum_{u_j \in I_i} d_{G_i} (u_j) \ge \frac{n_i}{136}$. 

This essentially yields $\Ex{X}  \ge \frac{1}{2} \cdot  \sum_{j =1}^{t} d_{G_i}(u_j) \ge \frac{n_i}{272}$ 
and we apply the Markov inequality to prove that $\Pr{|N(I_i)| \ge \frac{n_i}{900}} \ge 2/3$.  
\end{proof}

We can increase the success probability of Algorithm (1) {\sf Maximal-Independent-Set} to $1-\delta/n^3$ 
by creating $x = 3\log(n/\delta)$ runs $R_1,\cdots,R_x$ of this algorithm in parallel and report the MIS of 
the run $R_i$ whose neighborhood  size is at least $\frac{n_i}{900}$. 

Next we prove that at the end of a level $L_i$, for each vertex $v \in V_i$, either $v$ is deleted from 
the remaining graph $G_{i+1}$ or the degree of $v$ in $G_{i+1}$ is upper-bounded by $O(\log n\cdot \frac{m_i}{n_i})$. 

\begin{lemma}
\label{lem:high:degree:low:prob}
Let $L_i$ be a level in the level set $\mathcal{L}$. 
With probability at least $1-1/n^2$, each vertex $v \in V_i$ is either added to $I_i \cup N(I_i)$ and will not appear in $G_{i+1}$ or the degree of $v$ in the subgraph 
$G_{i+1}(V_{i+1},E_{i+1})$ is $d_{G_{i+1}}(v) \le \frac{3c\log n \cdot m_i}{n_i}$. 
\end{lemma}

\begin{proof}
Let us consider a graph $G_i(V_i,E_i)$ at a level $L_i$ where $n_i = |V_i|$ and $m_i = |E_i|$. 
In the beginning of the random sampling process  at level $L_i$, both $I_i$ and $N(I_i)$ are empty sets.  
Our sampling subroutine repeats the following process for $s = \frac{n_i^2}{cm_i}$ times: 
Repeat sampling (with replacement) a vertex $w \in V_i$ uniformly at random as long as $I_i \cup \{w\}$ is not an independent set in $G_i$. 
Once we sample a vertex $w$ for which $I_i \cup \{w\}$ is an independent set, 
we then let $I_i = I_i \cup \{w\}$ and $N(I_i) = N(I_i)  \cup N_{G_i}(w)$. 

Let us consider the process of building the independent set $I_i$ incrementally. 
That is, at each step $t \in [s]$, let $I_i^t = \{w_1,\cdots, w_{t}\}$ be an independent set that 
we found for $G_i$. Let $N(I_i^t)$ be the set of neighbors of $I_i^t$ till step $t$. 
Let $v \in V_i$ be a vertex with the neighbor set $N_{G_i}(v)$ and degree $d_{G_i}(w)$. 
Suppose $d_{G_i}(v) \ge \frac{3c\log n \cdot m_i}{n_i}$ as otherwise nothing left to prove. 

Let $X_v^t = N_{G_i}(v) \cup \{v\} \backslash ( I_i^t \cup N(I_i^t)$ be the set of neighbors of $v$ (including $v$) that are not 
in the independent set $I_i^t$ or adjacent to a vertex in $I_i^t$. 
Observe that if at step $t$ we sample a vertex $w \in X_v^t$, then $I_i^{t-1} \cup \{w\}$ will be an independent set and we can let $v_t = w$. 
If that happens, $v \in I_i \cup N(I_i)$ and the vertex $v$ is eliminated from $G_{i+1}$. So, suppose this does not happen. 
We then define a random event $\mathcal{BAD}_v^t$ for $|X_v^t| \ge \frac{3c\log n \cdot m_i}{n_i}$ but $v_t \notin X_v^t$ at step $t$. 

Observe that $\Pr{\mathcal{BAD}_v^t} \le 1- \frac{\frac{3c\log n \cdot m_i}{n_i}}{n_i} = 1- \frac{3\log n \cdot m_i}{n_i^2}$. 
Then, 
\[
\begin{split}
	&\Pr{ \mathcal{BAD}_v^1 \wedge \cdots \wedge \mathcal{BAD}_v^s} \\
	&= \Pr{\mathcal{BAD}_v^1}\cdot \Pr{\mathcal{BAD}_v^2 | \mathcal{BAD}_v^1} \cdot \Pr{\mathcal{BAD}_v^3 | \mathcal{BAD}_v^1 \wedge \mathcal{BAD}_v^2} \cdots 
	\Pr{\mathcal{BAD}_v^s | \mathcal{BAD}_v^1 \wedge \mathcal{BAD}_v^2 \wedge \cdots \wedge \mathcal{BAD}_v^{s-1}}\\
	&\le (1- \frac{3\log n \cdot m_i}{n_i^2})^s 
	= (1- \frac{3c\log n \cdot m_i}{n_i^2})^{\frac{n_i^2}{cm_i}} 
	\le e^{-3\log n} \le 1/n^3 \enspace .
\end{split}
\]

Using the union bound argument 
with probability at least $1-1/n^2$, each vertex $v \in V_i$ is either added to $I_i \cup N(I_i)$ and will not appear in $G_{i+1}$ or the degree of $v$ in the subgraph 
$G_{i+1}(V_{i+1},E_{i+1})$ is $d_{G_{i+1}}(v) \le \frac{3c\log n \cdot m_i}{n_i}$. 

\end{proof}


\section{Edge Insertion and Deletion}
\label{sec:insert:delete}
Here in this section we describe our edge insertion and deletion subroutines.


\begin{algorithm*}
\label{alg:insertion}
\noindent
\textbf{Edge-Insertion} (The level set $ \mathcal{L} = \cup_{i=1}^k L_i$ and an edge $e = (u,v)$)
\begin{algorithmic}[1]
	\IF{$u \in I_i$ and $v \in N(I_j)$ for $i,j \in [k]$ and $i < j$} 
		\STATE Invoke {\sf $\mathcal{L}$ = Light-Promotion $(\mathcal{L}, v, i)$}
	\ENDIF
	\IF{$u \in I_i$ and $v \in I_j$ for $i,j \in [k]$  and $i \le j$} 
		\STATE Invoke {\sf $\mathcal{L}$ = Light-Promotion $(\mathcal{L}, v, i)$}
		\STATE Invoke {\sf $\mathcal{L}$ = Heavy-Promotion $(\mathcal{L}, v)$}
	\ENDIF
\end{algorithmic}

\noindent\rule{16.5cm}{0.4pt}\\
\textbf{Edge-Deletion} (The level set $ \mathcal{L} = \cup_{i=1}^k L_i$ and an edge $e = (u,v)$)
\begin{algorithmic}[1]
	\IF{$u \in I_i$ and $v \in N(I_i) \backslash N(I_i \backslash \{u\})$} 
		\STATE Invoke {\sf $\mathcal{L}$ = Demotion$(\mathcal{L}, v, i)$}
	\ENDIF
\end{algorithmic}

\noindent\rule{16.5cm}{0.4pt}\\
\textbf{Light-Promotion} (The level set $ \mathcal{L} = \cup_{i=1}^k L_i$, a vertex $v$ and a level $r < L(v)$)
\begin{algorithmic}[1]
	\STATE Let $j = L(v)$ be the level of the vertex $v$. 
	\FOR{level $r < \ell \le j$}
		\item Let $V_{\ell} = V_{\ell} \backslash \{v\}$ where $V_{\ell}$ is the vertex set in the level $L_{\ell} \in \mathcal{L}$. 
	\ENDFOR
	\STATE $N(I_r) = N(I_r) \cup \{v\}$ and $N(I_j) = N(I_j) \backslash \{v\}$. 
\end{algorithmic}

\noindent\rule{16.5cm}{0.4pt}\\
\textbf{Heavy-Promotion} (The level set $ \mathcal{L} = \cup_{i=1}^k L_i$ and a vertex $v$ with level $j = L(v)$) 
\begin{algorithmic}[1]
	\STATE Let $F = N(I_j)\backslash N(I_j \backslash \{v\})$ be the neighbors of $N(I_j)$ that become free if we remove $v$ from $I_j$. 
	\STATE Let $I_j = I_j \backslash \{v\}$ be the independent set $I_j$ after removal of $v$.  

	\FOR {each vertex $w \in F$}
		\STATE Invoke {\sf $\mathcal{L}$ = Demotion$(\mathcal{L}, w, j)$}
	\ENDFOR

\end{algorithmic}

\noindent\rule{16.5cm}{0.4pt}\\
\textbf{Demotion} (The level set $ \mathcal{L} = \cup_{i=1}^k L_i$ and a vertex $w$ that is in $N(I_j)$ for $j \in [k]$)
\begin{algorithmic}[1]
	\STATE Let $P(w) = N_G(w) \cap \mathcal{I}$ be the set of neighbors of $w$ that are in MIS $\mathcal{I} = \cup_{i =1}^k I_i $. 
	\IF{ $P(w)$ is not empty } 
		\STATE Let $z \in P(w)$ be a vertex with the lowest level $L(z) \le \min_{x \in P(w)} L(x)$. 
		\STATE $N(I_j) = N(I_j) \backslash \{w\}$ and $N(I_{L(z)}) = N(I_{L(z}) \cup \{w\}$. 
	\ELSE
		\FOR{ $r$ in range $(j,k)$}
			\STATE  Sample $w$ with probability $\Pr{w} = \frac{n_r}{cm_r}$. 
			\IF{$w$ is sampled and $I_r = I_r \cup \{w\}$ is an independent set in $G_r$} 
				\STATE Let $I_r = I_r \cup \{w\}$. 
				\FOR{each vertex $z \in N_{G_r}(w)$}
					\STATE Invoke {\sf $\mathcal{L}$ = Light-Promotion $(\mathcal{L}, z, r)$}
				\ENDFOR
				\STATE Break the loop for $r$. 
			\ENDIF
		\ENDFOR
	\ENDIF
\end{algorithmic}

\textbf{Output:} Return the level set $ \mathcal{L} = \cup_{i=1}^k L_i$. 

\caption{Edge Insertion and Deletion Subroutines}
\end{algorithm*}

\subsection{Insertion and Deletion Subroutines}
Let us first consider the insertion of an edge $e=(u,v)$. 
The insertion of $e$ can trigger one of the following cases: 
 
\mbox{}

\InGrayMiddle{\mbox{}\\ 
\small
Let $\mathcal{L} = \cup_{i=1}^k L_i$ be a level set of a graph $G(V,E)$. 
Let $1 \le i \le j \le k$ be two level indices. 
An edge insertion $e = (u,v)$ triggers 
\begin{itemize}
\item \textbf{$(i \leftrightarrow j)$-Light Insertion}  if $u \in N(I_i)$ and either $v \in N(I_i)$ or $v \in V_j$.
\item \textbf{$(i \leftarrow  j)$-Light Promotion} if $u \in I_i$ and $v \in N(I_j)$.
\item \textbf{$(i \leftarrow  j)$-Heavy Promotion} if $u \in I_i$ and $v \in I_j$.
\end{itemize}
\mbox{}\\[-0.35in]}

\mbox{}

First suppose the insertion of an edge $e=(u,v)$ triggers a light insertion. 
That is, there exists $1 \le i \le j \le k$ for which $u \in N(I_i)$ and either $v \in N(I_i)$ or $v \in V_j$.
We then only need to add $e$ to the neighborhood of $u$ and $v$, i.e.,  $N_{G}(u)$ and $N_{G}(v)$, 
and add $e$ to the neighbor set $N(I_i)$. We also need to update the density $\frac{n_r}{m_r}$ for $1 \le r \le j$. 
The density update of each level is done automatically and we move it to the pseudocode of Algorithm (\ref{alg:dynamic}) {\sf Dynamic-MIS} 
for the sake of simplicity of insertion and deletion subroutines. 

Second suppose the insertion of an edge $e=(u,v)$ triggers a light promotion.
That is, there exists $1 \le i \le j \le k$ for which $u \in I_i$ and $v \in N(I_j)$.
We promote $v$ from the neighbor set $N(I_j)$ up to the neighbor set $N(I_i)$. 
We then eliminate $v$ from each vertex set $V_{\ell}$  for $i < \ell \le j$. 
Since $k = O(\log n)$, the light promotion subroutine takes $O(\log n)$ time. 

Finally, we consider the case when the insertion of an edge $e=(u,v)$ triggers a heavy promotion.
That is, there exists $1 \le i \le j \le k$ for which $u \in I_i$ and $v \in I_j$. 
The vertex $v$ is moved from $I_j$ to $N(I_i)$. By this operation, all neighbors of $v$ in $G_j$ that are not incident to 
any other vertex in $I_j\backslash \{v\}$ (that is, $w \in F = N(I_j) \backslash N(I_j \backslash \{v\})$) become free. 
For every such a vertex $w $ we demote $w$. 
That is, if there exists a vertex 
in one of independent sets $I_{r}$ for $r \ge j$ we demote $w$ to the level $L_r$ and add it to $N(I_r)$. 
Otherwise, we check to see if we can add $w$ to an independent set $I_{r}$ for $r \ge j$. 
In particular, for each level $L_r$ for $j \le r \le k$ with probability $\frac{n_r}{cm_r}$ and only if $I_r \cup \{w\}$ 
is an independent set in $G_r$ we add $w$ to $I_r$ and promote vertices in $N_{G_r}(w)$ to the level $L_r$ and add them to $N(I_r)$. 
Since $w$ is not adjacent to any vertex in an independent set $I_r$, the promotion of 
vertices $N_{G_r}(w)$ takes at most $d_{G_r}(w) \le d_{G_j}(w)$ time. 
%

As for the deletion of an arbitrary edge $e=(u,v)$, if $u \in I_i$ and $v \in N(I_i) \backslash N(I_i \backslash \{u\})$, 
we demote the vertex $v$. 
That is, if $v$ is adjacent to any independent set $I_{r \ge i}$, we demote $v$ to $N(I_r)$, otherwise we downsample 
$v$ with probability $\frac{n_r}{cm_r}$ for $r \ge j$ and check if we can add it to $I_r$ the same as edge insertion. 

Finally at any time $t$ if there exists a level $L_r$ whose density $\frac{m_r}{n_r}$ is increased or decreased by a factor of at least two, 
we recompute the maximal independent sets of all levels $L_{\ell \ge r}$. 
The density update of each level is done automatically and is moved to the pseudocode of Algorithm (\ref{alg:dynamic}) {\sf Dynamic-MIS}. 

\subsection{Analysis}

Let $\mathcal{L} = \cup_{i=1}^k L_i$ be a level set of an underlying graph $G(V,E)$. 
Recall that given $\mathcal{L}$, the reported maximal independent set is $\mathcal{I} = \cup_{i=1}^k I_i$. 
Let $e = (u,v) \in E$ be an arbitrary edge added to the graph $G$.

We first find an upper-bound for the probability that adding an arbitrary edge triggers a heavy promotion. 

\begin{lemma}
\label{lem:heavy:promotion:prob}
Let $1 \le i \le j \le k$ be two level indices. 
Let $c = 200$. 
The probability that adding an arbitrary edge $e=(u,v)$ triggers an $(i \leftarrow j)$-heavy promotion is at most $\frac{2}{c^2} \cdot \frac{n_i}{m_i} \cdot \frac{n_j}{m_j} $. 
That is, 
$$\Pr{u \in I_i \text{ and } v \in I_j }  \le  \frac{2}{c^2} \cdot \frac{n_i}{m_i} \cdot \frac{n_j}{m_j} = \frac{2}{c^2} \cdot d^{-1}(G_i) \cdot d^{-1}(G_j) \enspace ,$$
where $d^{-1}(G_i) $ and $ d^{-1}(G_j)$ are the multiplicative inverses or reciprocals for the average degree of $G_i$ and $G_j$, respectively. 
\end{lemma}

\begin{proof}
We define a random event $\mathcal{HP}_{i \leftarrow j}$ for $u \in I_i$ and $v \in I_j$.
We sample vertices in the graphs $G_i$ and $G_j$ with probabilities $\frac{n_i}{cm_i}$ and $\frac{n_j}{cm_j}$, respectively. 
Therefore, $\Ex{|I_i|} = \frac{n_i^2}{c m_i}$ and $\Ex{|I_j|} = \frac{n_j^2}{c m_j}$. 
Since $n_j \le n_i$, we then have 
\[
\begin{split}
	\Pr{\mathcal{H}_{i \leftarrow j}} = \frac{\frac{n_i^2}{c m_i} \cdot \frac{n_j^2}{c m_j}}{{n_i \choose 2}} 
	= \frac{2 \cdot \frac{n_i^2}{c m_i} \cdot \frac{n_j^2}{c m_j}}{n_i (n_i -1)} \le \frac{2n_i n_j}{c^2 m_i m_j} 
	= \frac{2}{c^2} \cdot \frac{n_i}{m_i} \cdot \frac{n_j}{m_j} 
	= \frac{2}{c^2} \cdot d^{-1}(G_i) \cdot d^{-1}(G_j) \enspace .
\end{split}
\]

We can also define a random event $\mathcal{HP}$ if there exist two indices $1 \le i \le j \le k$ for which we have $u \in I_i$ and $v \in I_j$. 
\[
\begin{split}
	\Pr{\mathcal{HP}} &= \frac{{|\mathcal{I}| \choose 1}}{{n \choose 2}} = \frac{|\mathcal{I}| \cdot (|\mathcal{I}|-1)}{n (n-1)} 
	\le \sum_{i \in [k]} \sum_{j \in [k]} \Pr{\mathcal{H}_{i \leftarrow j}} 
	\le \sum_{i \in [k]} \sum_{j \in [k]}  2\frac{n_i}{c m_i} \cdot \frac{n_j}{c m_j} \\
	& = \sum_{i \in [k]} \sum_{j \in [k]}   \frac{2}{c^2} \cdot d^{-1}(G_i) \cdot d^{-1}(G_j) \enspace .
\end{split}
\] 

\end{proof}

Next we bound the expected number of queries that our insertion and deletion subroutines need to recompute a maximal independent set 
after a light insertion, a light promotion or a heavy promotion happen. 

\begin{lemma}
\label{lem:light:insertion:cost}
Let $e=(u,v) \in E$ be an arbitrary edge.  
Let $c_{LI}$ be a large enough constant.  
Suppose that $e$ triggers a light insertion which happens if 
there exists $1 \le i \le j \le k$ for which $u \in N(I_i)$ and either $v \in N(I_i)$ or $v \in V_j$.
Then,  $\mathcal{Q}(\text{Light-Insertion}(v),G) = c_{LI} \log n $. 
\end{lemma}

\begin{proof}
If $e$ triggers a light insertion, we need to add $e$ to the neighborhood of $u$ and $v$, i.e.,  $N_{G}(u)$ and $N_{G}(v)$, 
and add $e$ to the neighbor set $N(I_i)$. We also need to update the density $\frac{n_r}{m_r}$ for $1 \le r \le j$. 
This can be done using three query and update operations plus $O(\log n)$ density updates. 
Thus, $\mathcal{Q}(\text{Light-Insertion}(v),G) = c_{LI}\cdot \log n $ for large enough constant $c_{LI}$. 
\end{proof}

\begin{lemma}
\label{lem:light:promotion:cost}
Let $e=(u,v) \in E$ be an arbitrary edge.  
Let $c_{LP}$ be a large enough constant.  
Suppose that $e$ triggers a light promotion that occurs when 
there exists $1 \le i \le j \le k$ for which $u \in I_i$ and $v \in N(I_j)$.
Then,  $\mathcal{Q}(\text{Light-Promotion}(v),G)  = c_{LP} \cdot \log n$. 
\end{lemma}

\begin{proof}
We promote $v$ from $N(I_j)$ up to $N(I_i)$. 
We then eliminate $v$ from each vertex set $V_{\ell}$  for $i < \ell \le j$. 
Since $k = O(\log n)$, the light promotion subroutine takes $\mathcal{Q}(\text{Light-Promotion}(v)),G)  = c_{LP} \cdot \log n$ time for large enough constant $c_{LI}$. 
\end{proof}

\begin{lemma}
\label{lem:heavy:promotion:cost}
Let $e=(u,v) \in E$ be an arbitrary edge.  
Let $c_{HP}$ be a large enough constant.  
Suppose that $e$ triggers a heavy promotion that occurs when 
there exists $1 \le i \le j \le k$ for which $u \in I_i$ and $v \in I_j$.
Then,  
$$
	\Ex{\mathcal{Q}(\text{Heavy-Promotion}(v),G)}  
	\le c_{HP}\log n \cdot \Ex{ \sum_{w \in N_{G_j}(v)} d_{G_j}(w)} 
	\le  c_{HP}\log n \cdot  \frac{3c\log n \cdot m_i}{n_i} \cdot \frac{m_j}{n_j}\enspace .
$$
\end{lemma}

%

\begin{proof}
The vertex $v$ is moved from $I_j$ to $N(I_i)$. By this operation, all neighbors of $v$ in $G_j$ that are not incident to 
any other vertex in $I_j\backslash \{v\}$ (that is, vertices in $F = N(I_j) \backslash N(I_j \backslash \{v\})$) become free. 
For each vertex $w \in F$ one of the following cases can happen. 

\textbf{Case 1:} If there exists a vertex in one of independent sets $I_{r}$ for $r \ge j$, we then demote $w$ to the level $L_r$ and add it to $N(I_r)$. 
To this end, we need to query the neighborhood of $v$ in $G_j$ which takes $O(d_{G_j}(v))$. 
Observe that $\Ex{d_{G_j}(v)} = \frac{m_i}{n_i}$. 
We also need to update the vertex sets $V_{\ell}$ and update the density $\frac{m_{\ell}}{n_{\ell}}$ for $j < \ell \le r$ what needs $O(\log n)$ query updates.\\

\textbf{Case 2:} Otherwise, for each level $L_r$ for $j \le r \le k$ with probability $\frac{n_r}{cm_r}$ and only if $I_r \cup \{w\}$ 
is an independent set in $G_r$ we add $w$ to $I_r$ and promote vertices in $N_{G_r}(w)$ to the level $L_r$ and add them to $N(I_r)$. 
Since $w$ is not adjacent to any vertex in an independent set $I_r$, the promotion of 
vertices $N_{G_r}(w)$ takes at most $d_{G_r}(w) \le d_{G_j}(w) \cdot O(\log n)$ time where we need to update 
the sets $V_{\ell > r}$ and the density $\frac{m_{\ell}}{n_{\ell}}$ as the vertices in $N_{G_r}(w)$ are promoted to the level $r$. 

Let us study the expected value of the random variable $X = \sum_{w \in N_{G_j}(v)} d_{G_j}(w)$. 
We consider two cases either $i<j$ or $i=j$. 

For the case when $i < j$ for every $w \in N_{G_j}(v)$ using Lemma \ref{lem:high:degree:low:prob} with probability at least $1-1/n^2$ 
we have  $d_{G_j}(w) <  \frac{3c\log n \cdot m_i}{n_i}$. 
Since the edge $e=(u,v)$ is chosen arbitrary, we have $\Ex{d_{G_j}(v)} = \frac{m_j}{n_j}$. 
Therefore, 
$$
\Ex{X} = \Ex{ \sum_{w \in N_{G_j}(v)} d_{G_j}(w)} \le  \frac{3c\log n \cdot m_i}{n_i} \cdot \Ex{d_{G_j}(v)} = \frac{3c\log n \cdot m_i}{n_i} \cdot \frac{m_j}{n_j} \enspace ,
$$
what yields 
$$
\mathcal{Q}(\text{Heavy-Promotion}(v)),G)  \le c_{HP}\log n \cdot  \Ex{\sum_{w \in N_{G_j}(v)} d_{G_j}(w)} = c_{HP}\log n \cdot \frac{3c\log n \cdot m_i}{n_i} \cdot \frac{m_j}{n_j} \enspace .
$$



The harder case is when $i=j$, especially when $i=j=1$ where we need to upper-bound the term $\Ex{ \sum_{w \in N_{G_j}(v)} d_{G_j}(w)}$. 
Observe that we can choose either $u$ or $v$.
So, the question boils down to study the expected sum of degrees of neighbors of a random vertex in $G_i$ 
for which we use Claim \ref{clm:bound:neighor:size:expectation} to show that 
$\Ex{\sum_{w \in N_{G_j}(v)} d_{G_j}(w)} = \frac{m_i}{n_i}\cdot \frac{m_i}{n_i}$ what proves this lemma. 

\begin{claim}
\label{clm:bound:neighor:size:expectation}
Let $v$ be a vertex that we sample uniformly at random from an independent set $I_i$ of a level $L_i$ for $i \in [k]$. 
Then, $\Ex{\sum_{w \in N_{G_j}(v)} d_{G_j}(w)} = \frac{m_i}{n_i}\cdot \frac{m_i}{n_i}$. 
\end{claim}

\begin{proof}
Let us define a random variable $X$ corresponding to the value $\sum_{w \in N_{G_j}(v)} d_{G_j}(w)$ of a random vertex $v \in I_i$. 
Then, we have 

\[
\begin{split}
	\Ex{X} &= \Ex{\sum_{w \in N_{G_j}(v)} d_{G_j}(w)} 
	= \sum_{v \in G_i} \Pr{v} d_{G_i}(v) \cdot \sum_{w \in G_i} \Pr{v \text{ incident to } w  \text{ in } G_i}\\  
	&= \sum_{v \in G_i} \Pr{v} d_{G_i}(v) \cdot \sum_{w \in G_i} \frac{d_{G_i}(w)}{n_i} 
	=  \sum_{v \in G_i}  \frac{{d_{G_i}(v)}}{n_i} \cdot \sum_{w \in G_i} \frac{d_{G_i}(w)}{n_i} = \frac{m_i}{n_i}\cdot \frac{m_i}{n_i} \\
\end{split}
\]

\end{proof}

\end{proof}

\begin{lemma}
\label{lem:bound:expected:query}
Let $e=(u,v)$ be an arbitrary edge  that is added to a graph $G(V,E)$ whose level set is $\mathcal{L} = \cup_{i=1}^k L_i$.
The expected number of queries that Algorithm {\sf Edge-Insertion} makes to update the level set $\mathcal{L}$ is 
$\Ex{\mathcal{Q}(\text{Edge-Insertion}(e=(u,v))),G)} = c\cdot \log^2 n$ where $c$ is a large enough constant. 
\end{lemma}

\begin{proof}
Suppose that  $u \in L_i$ and $v \in L_j$ for $i \le j$. 
We define a random variable $X$ for the number of queries that the insertion of an arbitrary edge $e = (u,v)$ causes. 
In the following we study the expectation of $X$. 

We define a random events $\mathcal{LI}_{i \leftarrow j}(e)$, $\mathcal{LP}_{i \leftarrow j}(e)$, and $\mathcal{HP}_{i \leftarrow j}(e)$ 
when the insertion of an arbitrary edge $e=(u,v)$ triggers a light insertion, a light promotion and a heavy promotion, respectively. 
Observe that $\Pr{\mathcal{LI}_{i \leftarrow j}(e )} + \Pr{\mathcal{LP}_{i \leftarrow j}(e )} +  \Pr{\mathcal{HP}_{i \leftarrow j}(e)} = 1$. 
Recall that using Lemmas \ref{lem:light:insertion:cost}, \ref{lem:light:promotion:cost}, \ref{lem:heavy:promotion:cost}, $\mathcal{Q}(\text{Light-Insertion}(v),G) = c_{LI}\cdot \log n$, 
$\mathcal{Q}(\text{Light-Promotion}(v),G)  = c_{LP} \cdot \log n$, and 
$$
\Ex{\mathcal{Q}(\text{Heavy-Promotion}(v),G)}  \le c_{HP}\log n \cdot \Ex{ \sum_{w \in N_{G_j}(v)} d_{G_j}(w)} =  c_{HP}\log n \cdot  \frac{3c\log n \cdot m_i}{n_i} \cdot \frac{m_j}{n_j}\enspace .
$$

\[
\begin{split}
	\Ex{X} &=  \Pr{\mathcal{LI}_{i \leftarrow j}(e )} \cdot \mathcal{Q}(\text{Light-Insertion}(v)),G)
		 + \Pr{\mathcal{LP}_{i \leftarrow j}(e )} \cdot \mathcal{Q}(\text{Light-Promotion}(v),G)\\ 
		&+ \Pr{\mathcal{HP}_{i \leftarrow j}(e)} \cdot \Ex{\mathcal{Q}(\text{Heavy-Promotion}(v),G)} \\
		&\le (c_{LI} + c_{LP})\cdot \log n + \Pr{\mathcal{HP}_{i \leftarrow j}(e)} \cdot c_{HP}\log n \cdot  \frac{3c\log n \cdot m_i}{n_i} \cdot \frac{m_j}{n_j}\\
		&\le (c_{LI} + c_{LP} + c_{HP})\log n\cdot ( 1  + \frac{6\log n}{c}) \le \frac{7(c_{LI} + c_{LP} + c_{HP})}{c}\cdot \log^2 n \enspace ,
\end{split}
\]
since $\Pr{\mathcal{LI}_{i \leftarrow j}(e )} + \Pr{\mathcal{LP}_{i \leftarrow j}(e )} \le 1$ and 
according to Lemma \ref{lem:heavy:promotion:prob} the event $\mathcal{HP}_{i \leftarrow j}$ occurs for an arbitrary edge $e=(u,v)$ 
with probability $\Pr{\mathcal{HP}_{i \leftarrow j}} \le \frac{2}{c^2} \cdot \frac{n_i}{m_i} \cdot \frac{n_j}{m_j}$.

\end{proof}

\begin{corollary}
\label{cor:bound:expected:query}
Let $e=(u,v)$ be an arbitrary edge  that is deleted from a graph $G(V,E)$ whose level set is $\mathcal{L} = \cup_{i=1}^k L_i$.
The expected number of queries that Algorithm {\sf Edge-Deletion} makes to update the level set $\mathcal{L}$ is 
$\Ex{\mathcal{Q}(\text{Edge-Deletion}(e=(u,v))),G)} = c\cdot \log^2 n$ where $c$ is a large enough constant.  
\end{corollary}

The proof of this corollary is the same as the proof of Lemma \ref{lem:bound:expected:query} and we omit it here.

\section{Analysis for Stream of Insertions and Deletions}
\label{sec:main:alg}
In this section we present our dynamic algorithm. 
The pseudocode of this algorithm is given in below. 
Lemma \ref{thm:mis:dynamic} proves that the amortized update time of this algorithm is $O(\log^3 n)$. 

\begin{algorithm*}
\label{alg:dynamic}
\noindent
\textbf{Input:}  A Sequence $\mathcal{S} = \{\textrm{Update}(e_1=(u_1,v_1)), \cdots, \textrm{Update}(e_z=(u_z,v_z))\}$ of edge updates to a graph $G$ 
where $\textrm{Update}(e_{\ell})$ is either $\textrm{Insertion}(e_{\ell})$ 
or $\textrm{Deletion}(e_{\ell})$.
\begin{algorithmic}[1]
	\STATE Let $G(V,E)$ be undirected unweighted graph whose vertex set $V$ of size $n=|V|$ is fixed and edge set $E$ that can be initialized to an empty set. 
	\STATE Initialize $y = 4\log(n/\delta)$ runs $R_1,\cdots,R_y$ in parallel.
	\FOR{$R_r$ where $r \in [y]$ in parallel}
		\STATE Invoke Algorithm (\ref{alg:offline:mis}) {\sf Maximal-Independent-Set$(G(V,E))$} whose output is a level set $\mathcal{L}^r = \cup_{i=1}^k L_i$. 
		\FOR{each update $\textrm{Update}(e_{\ell})$}
			\IF{$\textrm{Update}(e_{\ell})$ is $\textrm{Insertion}(e_{\ell})$}
				\STATE Invoke {\sf $\mathcal{L}^r = $ Edge-Insertion$(e_{\ell})$}. 
			\ELSE
				\STATE Invoke {\sf $\mathcal{L}^r = $ Edge-Deletion$(e_{\ell})$}. 
			\ENDIF
			\IF{there exists a level $L_i$ whose density $\frac{m_i}{n_i}$ changes by a factor $2$}
				\STATE Invoke Algorithm (\ref{alg:offline:mis}) {\sf Maximal-Independent-Set$(G_i(V_i,E_))$} that updates level set $\mathcal{L}^r = \cup_{i=1}^k L_i$.
			\ENDIF
		\ENDFOR
		\IF{the number of queries $\mathcal{Q}(\mathcal{A},G)$ that the run $R_r$ made up to now is greater than $3cz\cdot \log^3 n$}
			\STATE Stop the run $R_r$. 
		\ENDIF
	\ENDFOR
\end{algorithmic}
\noindent
\textbf{Output:} At any time $t\in [z]$, report the MIS maintained by a level set $\mathcal{L}^r = \cup_{i=1}^k L_i$ whose run $R_r$ survives. 
\caption{Dynamic-MIS}
\end{algorithm*}

\begin{theorem}
\label{thm:mis:dynamic}
Let $G=(V,E)$ be an undirected unweighted graph whose level set is $\mathcal{L} = \cup_{i=1}^k L_i$.
Let $\mathcal{S} = \{\textrm{Update}(e_1=(u_1,v_1)), \cdots, \textrm{Update}(e_z=(u_z,v_z))\}$ be a sequence of edge updates to a graph $G$ 
where $\textrm{Update}(e_{\ell}=(u_{\ell},v_{\ell}))$ is either $\textrm{Insertion}(e_{\ell}=(u_{\ell},v_{\ell}))$ 
or $\textrm{Deletion}(e_{\ell}=(u_{\ell},v_{\ell}))$.
Let $0 < \delta <1$ be a parameter. 
There is a randomized algorithm that with probability at least $1-\delta/n^2$, applies this sequence of updates to $G$ and 
updates the level set $\mathcal{L}$ in time $O(z\cdot \log^3 n)$. 
That is, the amortized update time of this algorithm is  $O(\log^3 n)$. 
\end{theorem}

\begin{proof}
Let us define $z$ random variables $X_1, \cdots, X_z$ corresponding to these edge updates  
where $X_{\ell}$ corresponds to the number of queries that the update of the edge $e_{\ell} = (u_{\ell},v_{\ell})$ needs  
to update the level set $\mathcal{L} = \cup_{i=1}^k L_i$. 
From Lemma \ref{lem:bound:expected:query} and Corollary \ref{cor:bound:expected:query}, we have 
$$
\Ex{X_{\ell}} \le \max(\Ex{\mathcal{Q}(\text{Edge-Insertion}(e_{\ell})),G)} , \Ex{\mathcal{Q}(\text{Edge-Deletion}(e_{\ell} )),G)}) = c\cdot \log^2 n \enspace. 
$$

Let $X = \sum_{\ell =1}^z X_{\ell}$. We then have $\Ex{X} = cz\cdot \log^2 n$. 
Using Markov Inequality, $$\Pr{X \ge 3cz\cdot \log^2 n} \le 1/3 \enspace .$$

Next we increase the probability of correctness to $1-\delta/n^3$. 
For the sequence $\textrm{Update}(e_1=(u_1,v_1)), \cdots, \textrm{Update}(e_z=(u_z,v_z))$ of edge updates, we run $y = 4\log(n/\delta)$ 
instances of Algorithms {\sf Edge-Insertion} and {\sf Edge-Deletion}  in parallel. Let $R_1,\cdots, R_y$ be the set of these $y$ runs.  
At any time $1 \le t \le z$, if we observe that the sum of the number of queries that a run $R_r$ makes from the beginning of the sequence (i.e., time $1$)
up to $t$ is greater than $3cz\cdot \log^2 n$, we stop the run $R_r$. 

Let $Y_r$ corresponds to the run $R_r$ such that $Y_r=1$ if for the $r$-th run 
the sum of the number of queries that $R_r$ makes from time $1$ to $t$ is is greater than $3cz\cdot \log^2 n$, and $Y_r=0$ otherwise. 
Therefore, $\Ex{Y_r} = p \le 1/3$. Let $a=1/2$ and $Y=\sum_{r=1}^y Y_r$. Using additive Chernoff Bound \ref{lem:cher} we then have, 
\[
	\Pr{Y \ge y/2} \le \left[ (\frac{p}{a})^a\ \cdot \left(\frac{1-p}{1-a}\right)^{1-a}\right]^{y} \le \left[ \sqrt{2/3} \cdot \sqrt{\frac{2/3}{1/2}}\right]^y \le (\sqrt{8/9})^y \le \delta/n^4 \enspace ,
\]
for $y \ge 4\log_{\sqrt{9/8}}(n/\delta) $. Be the relation between logarithms we then have $y \ge 4\log(n/\delta)$. 

We can assume that $z \le n(n+1)/2=n^2/2+n/2 \le n^2$. Since after every $n^2$ updates we re-run the MIS algorithm (i.e., Algorithm \ref{alg:offline:mis}) 
from the beginning. Therefore, using a union bound, with probability at least $1-\delta/n^2$ after every update there exists at least 
one run that survives. 
\end{proof}


\newcommand{\Proc}{Proceedings of the~}

\newcommand{\STOC}{Annual ACM Symposium on Theory of Computing (STOC)}
\newcommand{\FOCS}{IEEE Symposium on Foundations of Computer Science (FOCS)}
\newcommand{\SODA}{Annual ACM-SIAM Symposium on Discrete Algorithms (SODA)}
\newcommand{\SOCG}{Annual Symposium on Computational Geometry (SoCG)}
\newcommand{\ICALP}{Annual International Colloquium on Automata, Languages and Programming (ICALP)}
\newcommand{\ESA}{Annual European Symposium on Algorithms (ESA)}
\newcommand{\CCC}{Annual IEEE Conference on Computational Complexity (CCC)}
\newcommand{\RANDOM}{International Workshop on Randomization and Approximation Techniques in Computer Science (RANDOM)}
\newcommand{\APPROX}{International Workshop on  Approximation Algorithms for Combinatorial Optimization Problems  (APPROX)}
\newcommand{\PODS}{ACM SIGMOD Symposium on Principles of Database Systems (PODS)}
\newcommand{\SSDBM}{ International Conference on Scientific and Statistical Database Management (SSDBM)}
\newcommand{\ALENEX}{Workshop on Algorithm Engineering and Experiments (ALENEX)}
\newcommand{\BEATCS}{Bulletin of the European Association for Theoretical Computer Science (BEATCS)}
\newcommand{\CCCG}{Canadian Conference on Computational Geometry (CCCG)}
\newcommand{\CIAC}{Italian Conference on Algorithms and Complexity (CIAC)}
\newcommand{\COCOON}{Annual International Computing Combinatorics Conference (COCOON)}
\newcommand{\COLT}{Annual Conference on Learning Theory (COLT)}
\newcommand{\COMPGEOM}{Annual ACM Symposium on Computational Geometry}
\newcommand{\DCGEOM}{Discrete \& Computational Geometry}
\newcommand{\DISC}{International Symposium on Distributed Computing (DISC)}
\newcommand{\ECCC}{Electronic Colloquium on Computational Complexity (ECCC)}
\newcommand{\FSTTCS}{Foundations of Software Technology and Theoretical Computer Science (FSTTCS)}
\newcommand{\ICCCN}{IEEE International Conference on Computer Communications and Networks (ICCCN)}
\newcommand{\ICDCS}{International Conference on Distributed Computing Systems (ICDCS)}
\newcommand{\VLDB}{ International Conference on Very Large Data Bases (VLDB)}
\newcommand{\IJCGA}{International Journal of Computational Geometry and Applications}
\newcommand{\INFOCOM}{IEEE INFOCOM}
\newcommand{\IPCO}{International Integer Programming and Combinatorial Optimization Conference (IPCO)}
\newcommand{\ISAAC}{International Symposium on Algorithms and Computation (ISAAC)}
\newcommand{\ISTCS}{Israel Symposium on Theory of Computing and Systems (ISTCS)}
\newcommand{\JACM}{Journal of the ACM}
\newcommand{\LNCS}{Lecture Notes in Computer Science}
\newcommand{\RSA}{Random Structures and Algorithms}
\newcommand{\SPAA}{Annual ACM Symposium on Parallel Algorithms and Architectures (SPAA)}
\newcommand{\STACS}{Annual Symposium on Theoretical Aspects of Computer Science (STACS)}
\newcommand{\SWAT}{Scandinavian Workshop on Algorithm Theory (SWAT)}
\newcommand{\TALG}{ACM Transactions on Algorithms}
\newcommand{\UAI}{Conference on Uncertainty in Artificial Intelligence (UAI)}
\newcommand{\WADS}{Workshop on Algorithms and Data Structures (WADS)}
\newcommand{\SICOMP}{SIAM Journal on Computing}
\newcommand{\JCSS}{Journal of Computer and System Sciences}
\newcommand{\JASIS}{Journal of the American society for information science}
\newcommand{\PMS}{ Philosophical Magazine Series}
\newcommand{\ML}{Machine Learning}
\newcommand{\DCG}{Discrete and Computational Geometry}
\newcommand{\TODS}{ACM Transactions on Database Systems (TODS)}
\newcommand{\PHREV}{Physical Review E}
\newcommand{\NATS}{National Academy of Sciences}
\newcommand{\MPHy}{Reviews of Modern Physics}
\newcommand{\NRG}{Nature Reviews : Genetics}
\newcommand{\BullAMS}{Bulletin (New Series) of the American Mathematical Society}
\newcommand{\AMSM}{The American Mathematical Monthly}
\newcommand{\JAM}{SIAM Journal on Applied Mathematics}
\newcommand{\JDM}{SIAM Journal of  Discrete Math}
\newcommand{\JASM}{Journal of the American Statistical Association}
\newcommand{\AMS}{Annals of Mathematical Statistics}
\newcommand{\JALG}{Journal of Algorithms}
\newcommand{\TIT}{IEEE Transactions on Information Theory}
\newcommand{\CM}{Contemporary Mathematics}
\newcommand{\JC}{Journal of Complexity}
\newcommand{\TSE}{IEEE Transactions on Software Engineering}
\newcommand{\TNDE}{IEEE Transactions on Knowledge and Data Engineering}
\newcommand{\JIC}{Journal Information and Computation}
\newcommand{\ToC}{Theory of Computing}
\newcommand{\MST}{Mathematical Systems Theory}
\newcommand{\Com}{Combinatorica}
\newcommand{\NC}{Neural Computation}
\newcommand{\TAP}{The Annals of Probability}
\newcommand{\TCS}{Theoretical Computer Science}
\newcommand{\IPL}{Information Processing Letter}
\newcommand{\Algorithmica}{Algorithmica}

\nocite{*}

\begin{thebibliography}{1}

\bibitem{DBLP:conf/stoc/AssadiOSS18}
Sepehr Assadi, Krzysztof Onak, Baruch Schieber, and Shay Solomon.
\newblock Fully dynamic maximal independent set with sublinear update time.
\newblock In {\em Proceedings of the 50th Annual {ACM} {SIGACT} Symposium on
  Theory of Computing, {STOC} 2018, Los Angeles, CA, USA, June 25-29, 2018},
  pages 815--826, 2018.

\bibitem{DBLP:conf/soda/AssadiOSS19}
Sepehr Assadi, Krzysztof Onak, Baruch Schieber, and Shay Solomon.
\newblock Fully dynamic maximal independent set with sublinear in n update
  time.
\newblock In {\em Proceedings of the Thirtieth Annual {ACM-SIAM} Symposium on
  Discrete Algorithms, {SODA} 2019, San Diego, California, USA, January 6-9,
  2019}, pages 1919--1936, 2019.

\bibitem{Cher}
H.~Chernoff.
\newblock A measure of asymptotic efficiency for tests of a hypothesis based on
  the sum of observations.
\newblock {\em \AMS}, 23(4):493 -- 507, 1952.

\bibitem{DBLP:journals/corr/abs-1804-08908}
Yuhao Du and Hengjie Zhang.
\newblock Improved algorithms for fully dynamic maximal independent set.
\newblock {\em CoRR}, abs/1804.08908, 2018.

\bibitem{DBLP:journals/corr/abs-1804-01823}
Manoj Gupta and Shahbaz Khan.
\newblock Simple dynamic algorithms for maximal independent set and other
  problems.
\newblock {\em CoRR}, abs/1804.01823, 2018.

\end{thebibliography}

\end{document}